%% file: main.tex
\documentclass[11pt]{article} %

\usepackage{pgfplots}
\pgfplotsset{width=7cm,compat=1.9}
\usepackage{authblk}
\usepackage{booktabs}
\usepackage{float}
\usepackage[margin=1in]{geometry}
\usepackage[scr=rsfso]{mathalfa}
\usepackage{mdframed} %
\usepackage{microtype} %
\usepackage{mathpazo}
\usepackage{amsthm}
\usepackage{xspace}
\usepackage{comment}
\usepackage{lipsum}
\usepackage{amsfonts}
\usepackage{graphicx}
\usepackage{epstopdf}
\usepackage{algorithm}
\usepackage[noend]{algorithmic}
\usepackage{tikz}
\usepackage{caption}
\usepackage{amsmath}
\usepackage{hyperref}
\usepackage{amssymb}
\usepackage{dirtytalk}
\usepackage{subcaption}
\usetikzlibrary{decorations.pathmorphing}
\ifpdf
  \DeclareGraphicsExtensions{.eps,.pdf,.png,.jpg}
\else
  \DeclareGraphicsExtensions{.eps}
\fi

\newtheorem{theorem}{Theorem} % [section] %
\newtheorem{lemma}[theorem]{Lemma} %
\newtheorem{corollary}[theorem]{Corollary} %
\newtheorem{definition}[theorem]{Definition} %
\newtheorem{claim}[theorem]{Claim} %
\renewcommand{\cal}[1]{\mathcal{#1}}

\newcommand{\cvrp}{\textsc{Capacitated Vehicle Routing}\xspace}

\newcommand{\opt}{{\rm opt}}
\newcommand{\OPT}{{\rm OPT}}
\newcommand{\ssmall}{small}
\newcommand{\bbig}{big}
\newcommand{\radlower}{\mathcal{D}}
\newcommand{\rad}{\mathcal{D}'}
\newcommand{\eps}{\epsilon}
\newcommand{\cost}{cost}
\newcommand{\tsp}{\mathcal{A}}

\let\epsilon=\varepsilon %

\def\polylog{\mathop{\rm polylog}\nolimits}
\def\aux{\mathop{\rm G_{aux}}\nolimits}

\DeclareMathOperator*{\E}{\mathbb{E}}

\begin{document}

\title{Improved Approximations for CVRP with Unsplittable Demands}

\author[]{Zachary Friggstad\thanks{Supported by an NSERC Discovery Grant and NSERC Discovery Accelerator Supplement Award.} }
\author[]{Ramin Mousavi}
\author[]{Mirmahdi Rahgoshay}
\author[]{Mohammad R. Salavatipour \thanks{Supported by NSERC.}}

\affil[]{Department of Computing Science, University of Alberta, Edmonton, Canada.
\authorcr
  \{\tt zacharyf,mousavih,rahgosha,mrs\}@ualberta.ca}

\date{}

\maketitle

\input{Abstract}
%\sloppy
\input{Intro}

\input{Critical_Lemma}

\input{Simple_Alg}

\input{Rand_Alg}

\bibliographystyle{alpha}
\bibliography{mybibliography}

\end{document}

%% file: Abstract.tex
\begin{abstract}
In this paper, we present improved approximation algorithms for the (unsplittable) Capacitated Vehicle Routing Problem (CVRP) in general metrics.
In CVRP, introduced by Dantzig and Ramser (1959), we are given a set of points (clients) $V$ together with a depot $r$ in a metric space, with each $v\in V$ having a demand $d_v>0$, and a vehicle of bounded capacity $Q$. The goal is
to find a minimum cost collection of tours for the vehicle, each starting and ending at the depot, such that each client is visited at least once and the total demands of the clients in each tour is at most $Q$. In the unsplittable variant we study, the demand of a node must be served entirely by one tour.
We present two approximation algorithms for unsplittable CVRP: a combinatorial $(\alpha+1.75)$-approximation, where $\alpha$ is the approximation factor for the Traveling Salesman Problem, and an approximation algorithm based on LP rounding with approximation guarantee $\alpha+\ln(2) + \delta \approx 3.194 + \delta$ in $n^{O(1/\delta)}$ time. Both approximations can further be improved by a small amount when combined with recent work by Blauth, Traub, and Vygen (2021), who obtained an $(\alpha + 2\cdot (1 -\epsilon))$-approximation for unsplittable CVRP for some constant $\epsilon$ depending on $\alpha$ ($\epsilon > 1/3000$ for $\alpha = 1.5$).
\end{abstract}

%% file: Intro.tex
\section{Introduction}\label{intro}
%The Traveling Salesman Problem (TSP) and its variants are among the most well known and well studied problems in Combinatorial Optimization.
%Vehicle routing problems is a class of such problems.
Vehicle routing problems are among the most well known and well studied problems in Combinatorial Optimization.
The goal is generally to find cost-efficient delivery routes for delivering items from depots to clients in a network using vehicles. The Capacitated Vehicle Routing Problem (CVRP), introduced by Dantzig and Ramser in 1959 \cite{Dantzig}, generalizes the classic Traveling Salesman Problem and has numerous applications. In CVRP, we are given as input a complete graph $G=(V,E)$ with metric edge weights (also referred to as costs) $c(e)\in\mathbb{R}_{\ge 0}$,
a depot $r\in V$, and a vehicle with capacity $Q>0$, and wish to compute a minimum weight/cost collection of tours, each starting and ending at the depot and visiting at most $Q$ customers, whose union covers all the customers. In the more general setting, each node $v$ is given along with a demand $d(v)\in \mathbb{Z}_{\ge 1}$ and the goal is to find a set of tours of the minimum total cost, each of which includes $r$, such that the union of the tours covers the demand at every client and every tour serves at most $Q$ demand. 

There are three common versions of CVRP: \emph{unit}, \emph{splittable}, and \emph{unsplittable}. In the splittable variant, the demand of a node can be delivered using multiple tours so each tour must also specify how much demand it serves at each client\footnote{One can show using that restricting the demand served to each client by each tour to integer quantities does not change the optimum solution cost.}. However, in the unsplittable variant the entire demand of a client must be delivered by a single tour ({\em eg.} each demand is an indivisible good of a certain size). This obviously requires that $d_v\leq Q$ for all clients $v$. The unit demand case is a special case of the unsplittable case where every node has a unit demand, and the demand of a client must be delivered by a single tour.
It is easy to see that the splittable demand case can be reduced to the unit demand case in pseudo-polynomial time using multiple collocated clients of unit demands. However,
the unsplittable version is more challenging. For example, it contains the bin-packing problem as a special case; when all clients are have distance 1 from $r$ and distance 0 from each other.

CVRP has also been referred to as the $k$-tours problem \cite{Arora-Euclidean-PTAS,stoc/AsanoKTT97}.
Both the splittable and unsplittable versions admit constant factor approximation algorithms in polynomial-time. Haimovich and Kan \cite{Haimovich-Kan} showed that a heuristic, called iterative partitioning, yields an $(\alpha+1(1-1/Q))$-approximation
 for the unit demand case if one uses an $\alpha$-approximation for the Traveling Salesman Problem (TSP). A similar approach produces a
$2 + (1 - 2/Q)\alpha)$-approximation for the unsplittable variant \cite{ALTINKEMER1987149}. 
Despite their simplicity, these remained the best approximations for these two variants for over 35 years. Recently, Blauth et al. \cite{VygenIPCO} improved these approximations giving an $(\alpha + 2 \cdot (1 - \eps))$-approximation algorithm for unsplittable CVRP and a $(\alpha + 1 - \eps)$-approximation algorithm for unit demand CVRP and splittable CVRP where $\eps$ is a constant depending only on $\alpha$. For $\alpha = 3/2$, they showed $\eps > 1/3000$.
All the variants are APX-hard in general metric spaces \cite{Papadimitrio-Yannakakis}.

In this paper we make significant progress on improving the approximation guarantee for unsplittable CVRP. More specifically we present a simple combinatorial algorithm with ratio 3.25, and then a 3.194-approximation algorithm based on linear programming (LP). Our algorithms are completely independent of the improvements by Blauth et al. \cite{VygenIPCO}. By incorporating their approach, we can further improve both ratios by a small constant $\eps' > 0$. However, for the sake of simplicity we prefer to present our main results without factoring in this last improvement.

\begin{theorem}\label{thm:theorem325}
There is an approximation algorithm for the unsplittable CVRP with ratio $\alpha+1.75$, where $\alpha$ is
the best approximation ratio for TSP. 
\end{theorem}
The running time of this algorithm is dominated by computing two $\alpha$-approximate TSP tours and a minimum cost matching.
For example, using the simple (combinatorial) Christofides-Serdyukov 1.5-approximation
we get a combinatorial 3.25-approximation for unsplittable CVRP whose running time is dominated by computing $O(1)$ perfect matchings in graphs with $O(|V|)$ nodes.

If we allow greater running time, we can improve the approximation guarantee further by using linear programming.
\begin{theorem}\label{thm:LPtheorem}
For any $\delta>0$, there is an approximation algorithm for unsplittable CVRP with ratio $\ln(2)+\alpha+\frac{1}{1-\delta}$ and running time $n^{O(\frac{1}{\delta})}$, where $\alpha$ is
the best approximation ratio for TSP. 
\end{theorem}

Finally, we show how combining these two results with the approach in \cite{VygenIPCO} actually yields further improvements: a combinatorial $(\alpha+1.75-\eps')$-approximation and an LP-based $(\alpha+\ln(2)+\frac{1}{1-\delta}-\eps')$-approximation in time $n^{O(\frac{1}{\delta})}$, where $\eps' > 0$ is an absolute constant.

%%%%%%%%%%%%%%%%%%%%%%%%%%%%%%%%%%%%%%%%%%%%%%%%%%%%%%%%%%%%%%%%%%%%%%%%%%%%%
\subsection{Related Work}
CVRP captures classic TSP when $Q$, the vehicle capacity, is at least the total demand of all clients.
For general metrics, Haimovich and Kan \cite{Haimovich-Kan} considered a simple heuristic, called tour partitioning, which starts from a TSP tour and then splits it into tours of size at most $Q$ by making back-and-forth trips to $r$ at certain points along the TSP tour. They showed this gives a $(1 + (1 - 1/Q)\alpha)$-approximation for splittable CVRP, where $\alpha$ is the approximation ratio for TSP. Essentially the same algorithm yields a 
$(2 + (1 - 2/Q)\alpha)$-approximation for unsplittable CVRP \cite{ALTINKEMER1987149}. 
These stood as the best-known bounds until recently, when Blauth et al. \cite{VygenIPCO} showed that given a TSP approximation $\alpha$, there is an $\eps > 0$ such that there is an $(\alpha + 2 \cdot (1 - \eps))$-approximation algorithm for CVRP. For $\alpha = 3/2$, they showed $\eps > 1/3000$. They also describe a $(\alpha + 1 - \eps)$-approximation algorithm for unit demand CVRP and splittable CVRP.

For the case of trees, Labbé et al. \cite{Labbe-Mercure} showed splittable CVRP is NP-hard, and Golden et al. \cite{Golden-Wong} showed unsplittable version is hard to approximate better than 1.5. This is via a simple reduction from bin packing. For splittable CVRP (again on trees), Hamaguchi et al. \cite{Hamaguchi-Katoh} defined a lower bound for the cost of the optimal solution and gave a 1.5 approximation with respect to the lower bound. Asano et al. \cite{stoc/AsanoKTT97} improved the approximation to $(\sqrt{41} - 1)/4$ with respect to the same lower bound and also showed the existence of instances whose optimal cost is exactly 4/3 times the lower bound. Later, Becker \cite{Becker18} gave a 4/3-approximation with respect to the  lower bound. Becker and Paul \cite{Becker-Paul-Bricriteria} showed a $(1, 1+ \epsilon)$-bicriteria polynomial-time approximation scheme for splittable CVRP in trees, i.e. a PTAS but every tour serves at most $(1+\eps)Q$ demand. Recently, Jayaprakash and Salavatirpour \cite{JS22} presented a QPTAS for unit-demand CVRP for trees and more generally graphs of bounded treewidth, bounded doubling metrics, or bounded highway dimension. Even more recently, building upon ideas of \cite{Becker-Paul-Bricriteria} and \cite{JS22}, Mathieu and Zhou \cite{MZ21} have presented a PTAS for splittable CVRP on trees.

Das and Mathieu \cite{Das-Mathieu} gave a quasi-polynomial-time approximation scheme (QPTAS) for CVRP in the Euclidean plane ($\mathbb{R}^2$). A PTAS for when $Q$ is $O(\log n/\log \log n)$ or $Q$ is $\Omega(n)$ was shown by Asano et al. \cite{stoc/AsanoKTT97}. A PTAS for Euclidean plane $\mathbb{R}^2$ for moderately large values of $Q$, {\em i.e.} $Q \le 2^{\log^\delta n}$ where $\delta = \delta(\eps)$, was shown by Adamaszek et al \cite{AdamaszekCL09}, building on the work of Das and Mathieu \cite{Das-Mathieu}. For high dimensional Euclidean spaces $\mathbb{R}^d$, Khachay et al. \cite{Khachay-PTAS} showed a PTAS when $Q$ is $O(\log^{1/d}n)$. For graphs of bounded doubling dimension, Khachay et al. \cite{Khachay-moderatenumer} gave a QPTAS when the optimal number of tours is $\polylog(n)$ and Khachay et al.  \cite{Khachay-moderatecapacity} gave a QPTAS when $Q$ is $\polylog(n)$.

The next results we summarize are all for the case $Q = O(1)$. CVRP remains APX-hard in general metrics in this case but is polynomial-time solvable on trees. There exists a PTAS for CVRP in the Euclidean plane ($\mathbb{R}^2$) (again for when $Q$ is fixed) as shown by Khachay et al. \cite{Khachay-PTAS}. A PTAS for planar graphs was given by Becker et al. \cite{PlanarPTAS-Klein} and a QPTAS for planar and bounded-genus graphs was then given by Becker et al. \cite{QPTAS-Planar-boundedgenus}.  A PTAS for graphs of bounded highway dimension and an exact algorithm for graphs with treewidth $tw$ with running time $O(n^{tw\cdot Q})$ was shown by Becker et al. \cite{Becker-boundedhighway}.  Cohen-Addad et al. \cite{Klein-minorfree} showed an efficient PTAS for graphs of bounded-treewidth, an efficient PTAS for bounded highway dimension, an efficient PTAS for bounded genus metrics and a QPTAS for minor-free metrics. %Again, note that these results are all under the assumption that $Q$ is fixed.

{\bf Organization of the paper:} We start with definitions and preliminaries in Section \ref{sec: critical lemma}.
The proof of Theorem \ref{thm:theorem325} is presented in Section \ref{sec: simple alg} and the proof of Theorem \ref{thm:LPtheorem} is presented in Section \ref{sec: rand alg}. Comments on incorporating our ideas with those in \cite{VygenIPCO}
appear at the end of Section \ref{sec: rand alg}.
% In Section \ref{sec: simple alg} we prove Theorem \ref{thm:theorem325} and in Section \ref{sec: rand alg} we prove Theorem \ref{thm:LPtheorem}.

%% file: Critical_Lemma.tex
\section{Preliminaries}\label{sec: critical lemma}
%We start by some definitions and lemmas that are used in both algorithms.
For ease of exposition, we assume we have scaled all the demands and the capacity of the vehicle so that the capacity is $1$ and each $d(v)\in (0,1]$ (so demands can be rational numbers). Also, we treat $r$ as a separate node from the rest of the nodes. Formally:
\begin{definition}[\cvrp]\label{CVRP def}
An instance $(V,r,c,d)$ of \cvrp (CVRP) consists of:
\begin{itemize}
    \item a set of clients $V$, where $|V|=n$,
    \item a depot $r$, not in $V$,
    \item metric travel costs/distances $c:(V\cup\{r\})\times( V\cup\{r\})\to\mathbb{R}_{\geq 0}$,
    \item a demand $d_v\in (0,1]$ for each customer $v\in V$. 
\end{itemize}
A feasible solution is a collection of tours $\mathcal{T}$ such that
\begin{itemize}
    \item every tour $T\in\mathcal{T}$ is a cycle containing $r$,
    \item every client belongs to exactly one tour,
    \item $\sum\limits_{v\in T}d_v\leq 1$ for all $T\in\mathcal{T}$.
\end{itemize}
The goal is to find a feasible solution with minimum cost where the cost is the sum of costs of the edges in the solution and denoted by $c(\mathcal{T}):=\sum\limits_{T\in\mathcal{T}}c(T):=\sum\limits_{T\in\mathcal{T}}\sum\limits_{(u,v)\in T}c(u,v)$
\end{definition}
Observe we are viewing a tour $T$ as both a set of edges comprising a cycle plus the set of endpoints of these edges, so we may use notation like $v \in T$ for a location $v$ and also $(u,v) \in T$ for a pair of locations $(u,v)$ appearing consecutively along the tour $T$. It is convenient to view the depot $r$ as having $d_r = 0$, for example when we sum the demand of all locations on a tour.

Fix an unsplittable CVRP instance ${\cal I}=(V,r,c,d)$ for the rest of this paper.
We use $\OPT$ to denote an optimal solution for $\cal I$ and $\opt$ the value of this
optimal solution. %We use the following definitions and notation throughout the paper.
\begin{definition}[Feasible tours]
A tour $T$ that spans $r$ and some clients is called feasible for $\cal I$ if the total demand of the clients in $T$ is at most $1$, i.e., $\sum\limits_{v\in T}d_v\leq 1$.
\end{definition}

Clients are partitioned into {\em small} and {\em big} clients based on a parameter $\delta\in [0,\frac{1}{2}]$, which will be chosen differently for our two algorithms. % which will be fixed in each subsequent section.

\begin{definition}[Small and big clients]\label{small and big def}
For a fixed $\delta\in[0,\frac{1}{2}]$, we say a client $v$ is small if $d_v\in [0,\delta]$, and big otherwise.
\end{definition}

%\begin{definition}[Some useful sums]\label{two sums}
Let $\radlower:=\sum\limits_{v\in V}2\cdot d_v\cdot c(r,v)$. 
This is historically referred to as the radial lower bound and the following
simple well-known lemma has been used often in previous work.

%{\bf ZF: Include citation for lemma?}
\begin{lemma}[Haimovich and Kan \cite{Haimovich-Kan}]\label{radial lowerbound}
$\radlower\leq \opt$.
\end{lemma}
\begin{proof}
Let $\mathcal{T}^*$ be a collection of optimal tours. For every vertex in $T\in\mathcal{T}^*$, we have $2\cdot d_v\cdot c(r,v)\leq d_v\cdot c(T)$ by the triangle inequality. Hence, for a tour $T\in\mathcal{T}^*$, we have $\sum\limits_{v\in T}2\cdot d_v\cdot c(r,v)\leq c(T)\cdot\sum\limits_{v\in T}d_v\leq c(T)$. The lemma follows by summing the previous inequality for each tour in $\mathcal{T}^*$.
\end{proof}

We also define a similar sum for small and big clients separately, i.e.,
$\radlower_{\ssmall}:=\sum\limits_{\substack{v\in V:\\ v~is~\ssmall}}2\cdot d_v\cdot c(r,v)$, and
$\radlower_{\bbig}:=\sum\limits_{\substack{v\in V:\\ v~is~\bbig}}2\cdot d_v\cdot c(r,v)$.
Also define
$\rad_{\bbig}:=\sum\limits_{\substack{v\in V:\\ v~is~\bbig}}2\cdot c(r,v)$, which is the cost of serving
all big clients using a separate tour for each client.
%Similar to $\radlower$, we define $\rad_{\ssmall}$ and $\rad_{\bbig}$.
%\end{definition}

Given a TSP tour, the algorithm by Haimovich and Kan has the vehicle begin by randomily filling the ``tank'' of demand it carries with some value $\theta \sim (0,1]$. It then travels about the TSP tour: if the tank has insufficient demand to serve a client it travels to the depot to get enough demand to serve the client, returns to serve the client, and then returns to the depot to refill the tank appropriately before resuming the tour. The probability that such a resupply trip is performed when trying to serve a client $v$ is $d_v$, so the total cost of performing these round trips is at most $2 \cdot \radlower \leq 2 \cdot \opt$ in expectation.

One of the main driving forces behind our improvements is the following idea. For a small client, if we think of the vehicle's tank as only holding $1-\delta$ demand and keep a reserved tank holding demand $\delta$, then if we cannot serve a client with the demand in the main tank, we can serve it using the reserve tank and only make one round trip to the depot to refill both tanks before proceeding. Both of our main algorithms balance this idea with approaches to handling big clients.

We formalize this notion of using a reserve tank in Lemma \ref{tour partitioning} below. When $\delta=0$ this gives the same result as in \cite{Haimovich-Kan,ALTINKEMER1987149}.

\begin{lemma}[$\delta$-tank lemma]\label{tour partitioning}
Let $\tsp$ be a TSP tour on $V\cup\{r\}$ and define small and big clients based on a fixed $\delta \leq 1/2$. There is an algorithm that turns $\tsp$ into a feasible solution for the CVRP instance with cost 
\begin{equation}\label{tour part. cost}
    c(\tsp)+\frac{1}{1-\delta}\cdot\radlower_{\ssmall}+\frac{2}{1-\delta}\cdot\radlower_{\bbig}-\frac{\delta}{1-\delta}\cdot\rad_{\bbig},
\end{equation}
and running time $O(n^2)$.
\end{lemma}
\begin{proof}
Number the clients $V=\{v_1,...,v_n\}$ in the order they appear in $\tsp$. Choose $\theta\in [0,1-\delta]$ uniformly at random. We tile the nonnegative real number line from $0$ to $\sum\limits_{v\in V}d_v$ where the first tile has length $\theta$ and all subsequent intervals have length $1-\delta$, i.e., the endpoints of the tiles are $\theta+\eta\cdot(1-\delta)$ for integers $\eta\geq 0$. For any client $v_i$ if there is an integer $\ell\geq 0$ where
\begin{equation}\label{delta tank: eq}
    \sum\limits_{j=1}^{i-1}d_{v_j}<\theta+\ell\cdot(1-\delta)\leq \sum\limits_{j=1}^{i}d_{v_j}
\end{equation}
then we call $v_i$ a \textbf{bad} client, see Figure \ref{figure}. If $i = 1$, we take the LHS sum in \eqref{delta tank: eq} to be 0. For each bad client $v_i$, do the following:
\begin{itemize}
    \item[i.] if $v_i$ is a small client, then add two copies of $(r,v_i)$; so the added cost here is $2\cdot c(r,v_i)$. This means make a single round trip from $v_i$ to the depot.
    \item[ii.] if $v_i$ is a big client do the following:
    \begin{itemize}
        \item[a.] if for some $\ell$, we have $\sum\limits_{j=1}^{i-1}d_{v_j}<\theta+\ell\cdot(1-\delta)<\sum\limits_{j=1}^{i}d_{v_j} - \delta $ then add four copies of $(r,v_i)$; this means make two round trips to the depot. So the added cost here is $4\cdot c(r,v_i)$,
        \item[b.] otherwise add two copies of $(r,v_i)$ (i.e. do one round trip to $r$); so the added cost is $2\cdot c(r,v_i)$.
    \end{itemize}
\end{itemize}
We construct a collection of feasible tours using $\tsp$ and added edges; we will use shortcutting in some steps but the triangle inequality means the cost of our resulting solution is at most the cost of $\tsp$ plus the cost of the edges added for bad clients.
We start from $r$ and follow the TSP tour $\tsp$.
When the next vertex to be visited, say $v_i$, is bad, we consider two cases. If we only added two copies of the edge $(r, v_i)$, then we complete the current tour by visiting $v_i$ and then going to $r$. If we added four copies of the edge $(r, v_i)$, we complete the tour instead by going from $v_{i-1}$ to $r$ and then serve $v_i$ with a tour that only visits $v_i$ and returns to $r$. In either case, we resume constructing the next tour by going from $r$ to $v_{i+1}$ and continuing along the TSP tour.
See Figure \ref{figure} for an illustration. It is simple to verify the total demand clients on each tour we output is at most $1$.

% When we reach to a bad vertex $v_i$ with two copies of $(r,v_i)$, we go from $v_i$ to $r$ and that would end the current tour. Then, from $r$ we go to $v_{i+1}$ to start the next tour and follow $\tsp$ again until we reach to another bad vertex (or the end of the $\tsp$). If we reach to a bad vertex $v_i$ with four copies of $(r,v_i)$, then we go from $v_{i-1}$ to $r$ to complete the current tour. Next, we visit $v_i$ and back to $r$ as a single tour and then we continue from $v_{i+1}$,  It is easy to see based on our construction, each tour is a feasible tour. {\bf RM: should we explain more why each tour is feasible? I think it's easy enough that reader can verify it but messy enough to mess up the write up.}

Next, we bound the expected cost of final solution. To do that we need the following two claims:

\begin{claim}\label{prob bad client}
$\Pr[v_i~is~a~bad~client]=\min\{1,\frac{d_{v_i}}{1-\delta}\}$.
\end{claim}
\begin{proof}
A client $v_i$ is a bad client exactly when at least one endpoint of a tile lands in $(\sum\limits_{j=1}^{i-1}d_{v_j},~\sum\limits_{j=1}^i d_{v_j}]$ which is exactly the claimed probability.
\end{proof}

\begin{claim}\label{prob two round trips}
For a big client $v_i$, we have $\Pr[4~copies~of~(r,v_i)~added]=\frac{d_{v_i}-\delta}{1-\delta}$.
\end{claim}
\begin{proof}
Similar to the previous proof, the event we add four copies of $(r,v_i)$ is exactly when an endpoint of a tile lies in $(\sum\limits_{j=1}^{i-1}d_{v_j},\sum\limits_{j=1}^i d_{v_j}-\delta]$. 
\end{proof}
The expected cost of the final solution is
\begin{equation}\label{delta tank: cost of partitioning}
\begin{aligned}
    c(\tsp)&+\sum\limits_{v~\ssmall}\Pr[v~is~a~bad~client]\cdot 2\cdot c(r,v)\\
    & + \sum\limits_{v~\bbig} 
    \big(\Pr[v~is~a~bad~client]-\Pr[4~copies~of~(r,v)~added]\big)\cdot 2\cdot c(r,v)\\
    & + \sum\limits_{v~\bbig} \Pr[4~copies~of~(r,v)~added]\cdot 4\cdot c(r,v)\\
    &\leq c(\tsp) + \frac{1}{1-\delta}\cdot\sum\limits_{v~\ssmall}2\cdot d_v\cdot c(r,v) + \sum\limits_{v~\bbig} \Pr[v~is~a~bad~client]\cdot 2\cdot c(r,v)\\
    & + \sum\limits_{v~\bbig} \Pr[4~copies~of~(r,v)~added]\cdot 2\cdot c(r,v)\\
    &\leq c(\tsp)+\frac{1}{1-\delta}\cdot\sum\limits_{v~\ssmall}2\cdot d_v\cdot c(r,v) + \frac{1}{1-\delta}\cdot\sum\limits_{v~\bbig} 2\cdot d_v\cdot c(r,v)\\
    & + \sum\limits_{v~\bbig} \frac{d_v-\delta}{1-\delta}\cdot 2\cdot c(r,v)\\
    & = c(\tsp) + \frac{1}{1-\delta}\cdot\radlower_{\ssmall} + \frac{2}{1-\delta}\cdot \radlower_{\bbig} - \frac{\delta}{1-\delta}\cdot \rad_{\bbig},
\end{aligned}
\end{equation}
where in the first and the second inequalities, we just substitute the probability terms with their values given by Claims \ref{prob bad client} \& \ref{prob two round trips} and the fact that $\min\{1,\frac{d_v}{1-\delta}\}\leq \frac{d_v}{1-\delta}$.

Note also that we can deterministically compute the optimal partitioning of the sequence $v_1, \ldots, v_n$ into consecutive subsequences corresponding to feasible tours using dynamic programming in $O(n^2)$ time. This will produce a solution with cost at most \eqref{delta tank: cost of partitioning} since that is a bound on the expected cost of a random partitioning of this sort.
% to compute the optimum partitioning of the sequence of nodes $v_1, \ldots, v_n$ into feasible tours, each of which is just a consecutive subsequence of this sequence.
% there are $O(n)$ different values of $\theta$ that would lead to different values of ., we might have different partitioning, so we can enumerate all of these $n$ different values and choose the one with the minimum cost and this cost is guaranteed to be at most the cost in \eqref{delta tank: cost of partitioning}.
\end{proof}

\begin{figure}
\centering
\begin{minipage}{.4\textwidth}
\centering
\scalebox{.6}{
\begin{tikzpicture}[scale=0.7]
        \draw (0,0) -- (12,0); %Axis
        \draw (0,.2) -- (0,-.2) {};
        \draw (1,.2) -- (1,-.2) {};
        \draw (3,.2) -- (3,-.2) {};
        \draw (4,.2) -- (4,-.2) {};
        \draw (5.5,.2) -- (5.5,-.2) {};
        \draw (6.5,.2) -- (6.5,-.2) {};
        \draw (9,.2) -- (9,-.2) {};
        \draw (9.8,.2) -- (9.8,-.2) {};
        \draw (10.6,.2) -- (10.6,-.2) {};
        \node[label=above:$0$] (f0) at (0,.2) {};
        \node[label=above:$v_1$] (f1) at (1,.2) {};
        \node[label=above:$v_2$] (f2) at (3,.2) {};
        \node[label=above:$v_3$] (f3) at (4,.2) {};
        \node[label=above:$v_4$] (f4) at (5.5,.2) {};
        \node[label=above:$v_5$] (f5) at (6.5,.2) {};
        \node[label=above:$v_6$] (f6) at (8.9,.2) {};
        \node[label=above:$v_7$] (f7) at (9.8,.2) {};
        \node[label=above:$v_8$] (f8) at (10.6,.2) {};
        
        \draw [decorate,decoration={brace,amplitude=5pt,mirror,raise=4ex}]
        (1,.3) -- (2.4,.3) node[midway,yshift=-3em]{$d_{v_2}-\delta$};
        \draw [decorate,decoration={brace,amplitude=5pt,mirror,raise=4ex}]
        (4,.3) -- (4.6,.3) node[midway,yshift=-3em]{$d_{v_4}-\delta$};
        \draw [decorate,decoration={brace,amplitude=5pt,mirror,raise=4ex}]
        (6.5,.3) -- (8.3,.3) node[midway,yshift=-3em]{$d_{v_6}-\delta$};
        
        \draw (0,-2) -- (12,-2); %Axis
        \draw[dashed] (0,0) -- (0,-2.2) {};
        \draw[dashed] (.8,0) -- (.8,-2.2) {};
        \draw[dashed] (2.8,0) -- (2.8,-2.2) {};
        \draw[dashed] (4.8,0) -- (4.8,-2.2) {};
        \draw[dashed] (6.8,0) -- (6.8,-2.2) {};
        \draw[dashed] (8.8,0) -- (8.8,-2.2) {};
        \draw[dashed] (10.8,0) -- (10.8,-2.2) {};
        
        \draw [decorate,decoration={brace,amplitude=5pt,mirror,raise=4ex}]
        (0,-1.5) -- (.8,-1.5) node[midway,yshift=-3em]{$\theta$};
        \draw [decorate,decoration={brace,amplitude=5pt,mirror,raise=4ex}]
        (0.8,-1.5) -- (2.8,-1.5) node[midway,yshift=-3em]{$1-\delta$};
        \draw [decorate,decoration={brace,amplitude=5pt,mirror,raise=4ex}]
        (2.8,-1.5) -- (4.8,-1.5) node[midway,yshift=-3em]{$1-\delta$};
        \draw [decorate,decoration={brace,amplitude=5pt,mirror,raise=4ex}]
        (4.8,-1.5) -- (6.8,-1.5) node[midway,yshift=-3em]{$1-\delta$};
        \draw [decorate,decoration={brace,amplitude=5pt,mirror,raise=4ex}]
        (6.8,-1.5) -- (8.8,-1.5) node[midway,yshift=-3em]{$1-\delta$};
        \draw [decorate,decoration={brace,amplitude=5pt,mirror,raise=4ex}]
        (8.8,-1.5) -- (10.8,-1.5) node[midway,yshift=-3em]{$1-\delta$};
\end{tikzpicture}
}
\subcaption{}
\end{minipage}
\begin{minipage}{.25\textwidth}
\centering
\scalebox{.6}{
\begin{tikzpicture}[scale=0.8,
smallnode/.style={circle, draw=black, fill=black, very thick, scale =0.2}]
\node[smallnode, label=below:$r$] (r) at (0,-3) {};
\node[smallnode, label=right:$v_1$] (v1) at (1.1,-2.7) {};
\node[smallnode, label=right:$v_2$] (v2) at (1.5,-2) {};
\node[smallnode, label=right:$v_3$] (v3) at (1.5,-1) {};
\node[smallnode, label=right:$v_4$] (v4) at (1,0) {};
\node[smallnode, label=left:$v_5$] (v5) at (-1,0) {};
\node[smallnode, label=left:$v_6$] (v6) at (-1.5,-1) {};
\node[smallnode, label=left:$v_7$] (v7) at (-1.3,-2) {};
\node[smallnode, label=left:$v_8$] (v8) at (-1.2,-2.6) {};

\draw[very thick] (r) -- (v1);
\draw[very thick] (v1) -- (v2);
\draw[very thick] (v2) -- (v3);
\draw[very thick] (v3) -- (v4);
\draw[very thick] (v4) -- (v5);
\draw[very thick] (v5) -- (v6);
\draw[very thick] (v6) -- (v7);
\draw[very thick] (v7) -- (v8);
\draw[very thick] (v8) -- (r);

\draw[dashed, very thick] (r) edge [bend right=40] node[midway, below]{$\times 2$} (v1);

\draw[dashed, very thick] (r) edge node[midway, above]{$\times 2$} (v2);

\draw[dashed, very thick] (r) edge node[midway, right]{$\times 2$} (v4);

\draw[dashed, very thick] (r) edge [bend right=30] node[pos=.7, right]{$\times 4$} (v6);

\end{tikzpicture}
}
\subcaption{}
\end{minipage}
\begin{minipage}{.25\textwidth}
\centering
\scalebox{.6}{
\begin{tikzpicture}[scale=.8,
smallnode/.style={circle, draw=black, fill=black, very thick, scale =0.2}]
\node[smallnode, label=below:$r$] (r) at (0,-3) {};
\node[smallnode, label=right:$v_1$] (v1) at (1.1,-2.7) {};
\node[smallnode, label=right:$v_2$] (v2) at (1.5,-2) {};
\node[smallnode, label=right:$v_3$] (v3) at (1.5,-1) {};
\node[smallnode, label=right:$v_4$] (v4) at (1,0) {};
\node[smallnode, label=left:$v_5$] (v5) at (-1,0) {};
\node[smallnode, label=left:$v_6$] (v6) at (-1.5,-1) {};
\node[smallnode, label=left:$v_7$] (v7) at (-1.3,-2) {};
\node[smallnode, label=left:$v_8$] (v8) at (-1.2,-2.6) {};

\draw[very thick] (r) -- (v1);
\draw[very thick] (r) edge [bend right=40] (v1);
\draw[very thick] (r) edge [bend left=20] (v2);
\draw[very thick] (r) -- (v3);
\draw[very thick] (v3) -- (v4);
%\draw[very thick] (v4) -- (v5);
\draw[very thick] (r) -- (v5);
\draw[very thick] (r) edge [bend right=30] (v5);
\draw[very thick] (r) -- (v6);
\draw[very thick] (v7) -- (v8);
\draw[very thick] (v8) -- (r);

\draw[very thick] (r) edge (v2);

\draw[very thick] (r) edge (v4);

\draw[very thick] (r) edge [bend right=20] (v6);

\draw[very thick] (r) edge [bend right=15] (v7);
\end{tikzpicture}
}
\subcaption{}
\end{minipage}
\caption{Illustrations of steps in the proof of Lemma \ref{tour partitioning}. Here, $v_2,v_4$, and $v_6$ are big clients and the rest are small clients. The random tiling is depicted in (a). According to this tiling, $v_1,v_2,v_4$, and $v_6$ are bad clients. In (b), the solid edges form a TSP tour $\tsp$ and the dashed edges are the added edges and the number on the dashed edges show how many copies of that edge is being added.
%For example, since the endpoint of the first tile lies in $[0,v_1]$, and $v_1$ is a small client, we have two copies of $(r,v_1)$. Note $v_2$ is a big client but since the endpoint of the tile lies outside of $[v_1,v_2-\delta]$, we only add two copies of $(r,v_2)$, same holds for $v_4$. Lastly, $v_6$ is a big client and there is an endpoint of a tile that lies within $[v_5,v_6-\delta]$ and therefore we add four copies of $(r,v_6)$.
In (c), we show how to transform $\tsp$ and the added edges into a collection of feasible tours for the CVRP instance with no more cost than the multiset of edges depicted in (b).
}
\label{figure}
\end{figure}
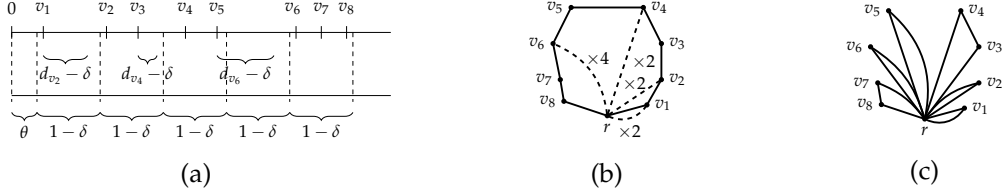

%% file: Simple_Alg.tex
\section{A Combinatorial $3.25$-Approximation}\label{sec: simple alg}
In this section, we set $\delta:=\frac{1}{3}$. So $v$ is a small client if $d_v\leq \frac{1}{3}$ and big if $d_v > \frac{1}{3}$. Note that in any feasible solution, there are at most two big clients in any single tour. Our algorithm tries two things: the first serves only big clients by pairing them up optimally to form these tours and then runs the classic $3.5$-approximation on the small clients but using our $\delta$-tank procedure (see Lemma \ref{tour partitioning}) for performing the tour splitting. The other simply runs the $3.5$-approximation using $\delta$-tank tour splitting on all clients.

Let us first explain how we use matching. Consider an auxiliary graph $\aux=(V_{\bbig}, E_{aux})$ where $V_{\bbig}\subseteq V$ is the set of all big clients and $E_{aux}$ constructed as follows: for any pair of big clients $u,v$ where $d_u+d_v\leq 1$ we add and edge between $u$ and $v$ with cost $c(r,u) + c(u,v) + c(v,r)$. Furthermore, for every big client $v$ there is a loop in $\aux$ with cost equal to $2\cdot c(r,v)$. We compute a min-cost perfect matching\footnote{A set of edges $M$ that may contain loops is a perfect matching if each node lies in precisely one edge: so a node is either matched with another node via a normal edge or with itself via a loop.} which corresponds to the cheapest way to select tours to serve only the big clients. The precise details are presented in Algorithm \ref{alg: 3.25}.

\begin{algorithm}[H]
  \caption{$(\alpha+1.75)$-approximation}\label{alg: 3.25}
\begin{algorithmic}[1]
\STATE The first solution is constructed as follows:
\begin{ALC@g}
   \STATE Compute a min-cost perfect matching $M$ on $\aux$. Let $\mathcal{T}'$ be the tours corresponding to the edges in $M$.
   \STATE Compute a TSP tour $\tsp$ on small clients and $r$.
   \STATE Apply Lemma \ref{tour partitioning} to $\tsp$ with $\delta=\frac{1}{3}$ and let $\mathcal{T}''$ be the resulting solution.
   \STATE $\mathcal{T}\gets \mathcal{T}'\cup \mathcal{T}''$
\end{ALC@g}
\STATE The second solution is constructed as follows:
\begin{ALC@g}
    \STATE Compute a TSP tour $\tsp$ on $V\cup\{r\}$.
    \STATE Apply Lemma \ref{tour partitioning} to $\tsp$ with $\delta=\frac{1}{3}$ and let $\mathcal{F}$ be the resulting solution.
\end{ALC@g}
\STATE Return the cheaper of the two solutions $\mathcal T$ and $\mathcal F$.%$\argmin\{c(\mathcal{T}),c(\mathcal{F})\}$.
\end{algorithmic}
\end{algorithm}

\subsection{Analysis}
We begin with two simple observations.
\begin{lemma}\label{3.25: matching cost by opt}
$\cost(M)\leq \opt$.
\end{lemma}
\begin{proof}
Each tour in any feasible solution contains at most two big clients.
So, after shortcutting all tours in $\OPT$ past small clients, we get tours corresponding to a perfect matching $\aux$ with cost at most $\opt$.
\end{proof}
\begin{lemma}\label{3.25: matching cost by D'}
$\cost(M)\leq \rad_{\bbig}$.
\end{lemma}
\begin{proof}
Consider all the loops in $\aux$. The cost of all the loops is exactly $\rad_{\bbig}$ and this is a matching so it is an upper bound on the minimum cost of a perfect matching.
\end{proof}

Next, we compute the cost of the first solution in the algorithm. Note that $c(\mathcal{T}')=\cost(M)$. Using an $\alpha$-approximation for TSP, the cost of $\tsp$ is at most $\alpha\cdot\opt$: again we are using the metric property which shows $\opt$ upper bounds the optimum TSP tour since the union of all tours in $\OPT$ is connected and Eulerian. Finally, applying the $\delta$-tank lemma to $\tsp$ results in a solution of cost at most $c(\tsp)+\frac{1}{1-\delta}\cdot\radlower_{\ssmall}$ since there is no big client on $\tsp$. Overall, we have
\begin{equation}\label{3.25: sol 1}
\begin{aligned}
    c(\mathcal{T})&= c(\mathcal{T}')+c(\mathcal{T}'')
    \leq \cost(M) + \alpha\cdot\opt + \frac{3}{2}\cdot\radlower_{\ssmall}.
\end{aligned}
\end{equation}

Next, we compute the cost of the second solution. From the $\delta$-tank lemma, %$c(\mathcal{F})$ is exactly the cost computed in the $\delta$-tank lemma.

\begin{equation}\label{3.25: sol 2}
    c(\mathcal{F})=\alpha\cdot \opt + \frac{3}{2}\cdot\radlower_{\ssmall} + 3\cdot\radlower_{\bbig} - \frac{1}{2}\cdot\rad_{\bbig}.
\end{equation}

%The cheapest of these two costs is at most their average: output by the solution is at most the average of \eqref{3.25: sol 1} and \eqref{3.25: sol 2} we get
Combining these, we bound the cost of the solution output by the algorithm as follows:
\begin{equation}\label{eq final approx factor}
\begin{aligned}
    \min\{c(\mathcal{T}),c(\mathcal{F})\}&\leq \frac{c(\mathcal{T})+c(\mathcal{F})}{2}\\
    & = \frac{2\cdot\alpha\cdot\opt + 3\cdot (\radlower_{\ssmall}+\radlower_{\bbig})+\cost(M)-\frac{1}{2}\cdot\rad_{\bbig}}{2}\\
    &\leq \frac{2\cdot\alpha\cdot\opt + 3\cdot \radlower+\frac{1}{2}\cdot \cost(M)}{2}\\
    &\leq \alpha\cdot\opt+1.5\cdot\opt+0.25\cdot\opt\\
    &=(\alpha+1.75)\cdot\opt,
\end{aligned}
\end{equation}
where the second inequality follows from Lemma \ref{3.25: matching cost by D'} and the last inequality follows from Lemmas \ref{radial lowerbound} \& \ref{3.25: matching cost by opt}. This finishes the proof of Theorem \ref{thm:theorem325}.

% \textbf{option 1:}
% We can further improve this approximation factor by $1.5\cdot\epsilon$ where $\epsilon=\frac{1}{3000}$ is the absolute constant in \cite{Vygen}. If $\radlower\leq (1-\epsilon)\cdot\opt$ then in \eqref{eq final approx factor} substituting this upper bound for $\radlower$ yields an upper bound of $(3.25-1.5\cdot\epsilon)\cdot\opt$ for our solution. On the other hand, if $\radlower\geq (1-\epsilon)\cdot\opt$, then this is the case that \cite{Vygen} calls {\em difficult}. In the difficult case, they provide a TSP tour of cost at most $(\frac{3}{2}-2\cdot\epsilon)\cdot\opt$. Substituting this upper bound for $\alpha$ in \eqref{eq final approx factor}, we have an upper bound of $(3.25-2\cdot\epsilon)\cdot\opt$ on our solution.

% \textbf{option 2:}
% We can further improve this approximation factor by some small but constant factor. If $\radlower\leq (1-\epsilon)\cdot\opt$ where $\epsilon=\frac{1}{3000}$ is the absolute constant in \cite{Vygen}, then \eqref{eq final approx factor} already yields an improvement by some constant factor of $\epsilon$. On the other hand, if $\radlower\geq (1-\epsilon)\cdot\opt$, then this is the case that \cite{Vygen} calls {\em difficult}. In this case, they get a TSP tour whose cost is better than $1.5\cdot\opt$ and the improvement is some constant factor of $\epsilon$. Using this improved bound for $\alpha$ in \eqref{eq final approx factor}, we get an upper bound better than $3.25\cdot\opt$ for our solution again by some constant factor of $\epsilon$.

%% file: Rand_Alg.tex
\section{An Improved LP-Based Approximation%$(3.194+\frac{1}{1-\delta})$-Approximation
}\label{sec: rand alg}
In this section let $\delta$ be a fixed constant in the range $(0, 1/2]$. Smaller $\delta$ lead to better approximations with increased, but still polynomial, running times.
%That is, for any $0<\delta\leq 1/2$, we obtain an approximation factor that depends on $\delta$ and the running time is $n^{O(\frac{1}{\delta})}$. 

Define the small and big clients for this value $\delta$ as in Definition \ref{small and big def}. Let $V_{\bbig}$ be the set of big clients. We consider the following configuration LP for big clients: Let $\mathcal{J}$ be the set of all feasible tours where each tour consists of some big clients and the depot. Note $|\mathcal{J}|$ is bounded by $n^{O(\frac{1}{\delta})}$ as there can be at most $\frac{1}{\delta}$ big clients in each tour. For each $T\in\mathcal{J}$ let $c(T)$ be the cost of tour $T$. For each tour $T\in\mathcal{J}$, we have a variable $x_{T}$ indicating this tour is chosen by the algorithm.

\begin{alignat}{3}\label{config lp}
\text{minimize:} & \quad & \sum_{T\in\mathcal{J}} c(T)\cdot x_T \tag{{\bf Configuration-LP}} \\
\text{subject to:} && \sum_{\substack{T\in \mathcal{J}:\\ v\in T}}x_T \geq \quad & 1 \quad && \forall v \in V_{\bbig} \\
&& x \geq \quad & 0 \notag
\end{alignat}
By shortcutting all tours in the optimum solution past small clients and discarding tours with no big clients, we see there is an integer solution to \eqref{config lp} with cost at most $\opt$. Thus, the optimum LP value provides a lower bound on $\opt$.

Our algorithm independently samples tours spanning large clients using an optimal LP solution. After this, some large clients and all small clients remain uncovered, we cover them using the classic $3.5$-approximation but use the $\delta$-tank tour splitting approach. Algorithm \ref{alg: 3.19} contains the full description of our approach. With foresight, we set $\gamma := \ln(2)$.

\begin{algorithm}[H]
  \caption{$(3.194+\frac{1}{1-\delta})$-approximation}\label{alg: 3.19}
\begin{algorithmic}[1]
    \STATE $\mathcal{T}\gets \emptyset$. \COMMENT{This will be a collection of tours.}
    \STATE Compute an optimal solution $x^*$ of \eqref{config lp}.
    \FOR{$T\in\mathcal{J}$}
    \STATE with probability $\min\{1,\gamma\cdot x_T\}$ add $T$ to $\mathcal T$.
    \ENDFOR
    \STATE Approximate a TSP tour $\tsp$ spanning $\{r\}\cup(V\setminus V(\mathcal{T}))$ where $V(\mathcal{T})$ is the vertices covered in $\mathcal{T}$.
    \STATE Apply the $\delta$-tank lemma to $\tsp$ and let $\mathcal{T}'$ be the resulting collection of tours.
    \STATE Return $\mathcal{T}\cup\mathcal{T}'$.
\end{algorithmic}
\end{algorithm}
It could be that some clients lie on multiple tours due to the randomized rounding step. One can shortcut the tours past repeated occurrences of clients so each client lies on exactly one tour.

\subsection{Analysis}

We first bound the probability of a big client not being covered in the randomized rounding step of Algorithm \ref{alg: 3.19} (steps 3-4).
\begin{lemma}\label{prob big not covered}
For a $v \in V_{\bbig}$, $\Pr[v~is~not~covered~by~\mathcal{T}]\leq e^{-\gamma}$.
\end{lemma}
\begin{proof}
The event that a big client $v$ is not covered is if we do not sample any tour $T$ that contains $v$ in the randomized rounding step. So 
\[
  \Pr[v~is~not~covered~by~\mathcal{T}]=\prod\limits_{T\in\mathcal{T}}(1-\gamma\cdot x_T)\leq e^{-\gamma\cdot\sum\limits_{T\in\mathcal{T}:v\in T}x_T}\leq e^{-\gamma},
\]
where the last bound follows from the constraint in \eqref{config lp} for $v$.
\end{proof}
Next, we bound the expected costs of $\mathcal{T}$ and $\mathcal{T}'$, separately. The cost of $\mathcal{T}$ is bounded as follows:
\begin{equation}\label{T cost}
    \E[\mathcal{T}]=\gamma\cdot\cost(x^*)\leq \gamma\cdot\opt.
\end{equation}
Using the $\delta$-tank lemma, we bound the expected cost of $\mathcal{T}'$ but with the following changes: in \eqref{tour part. cost}, we drop the negative term and we incorporate the fact that a big client is on $\tsp$ with probability at most $e^{-\gamma}$, see Lemma \ref{prob big not covered}.
\begin{equation}
\begin{aligned}\label{T' cost}
    \E[c(\mathcal{T}')]%&\leq c(\tsp)+\frac{1}{1-\delta}\cdot\radlower_{\ssmall} +\E[2\cdot\radlower_{\bbig}]\\
    & \leq c(\tsp)+\frac{1}{1-\delta}\cdot\radlower_{\ssmall}\\
    &+ 2\cdot\sum\limits_{v\in V_{\bbig}}\Pr[v~is~not~covered~by~\mathcal{T}]\cdot \frac{d_v}{1-\delta}\cdot 2\cdot c(r,v)\\
    &= c(\tsp)+\frac{1}{1-\delta}\cdot\radlower_{\ssmall} + e^{-\gamma}\cdot\frac{2}{1-\delta}\cdot\radlower_{\bbig}\\
% \end{aligned}
% \end{equation}
% With our choice of $\gamma = \ln 2$, we have 
% \begin{equation}\label{T' cost}
% \begin{aligned}
    %&= \E[c(\mathcal{T}')]
    &= c(\tsp)+\frac{1}{1-\delta}\cdot\radlower_{\ssmall}+\frac{1}{1-\delta}\cdot\radlower_{\bbig}\\
    &=c(\tsp)+\frac{1}{1-\delta}\cdot\radlower \leq \alpha\cdot \opt+\frac{1}{1-\delta}\cdot\radlower.
\end{aligned}
\end{equation}
The second equality follows from our choice of $\gamma = \ln 2$.
From \eqref{T cost} and \eqref{T' cost}, the expected cost of the solution returned by Algorithm \ref{alg: 3.19} is at most
\begin{equation}
\begin{aligned}
    \E[c(\mathcal{T}\cup\mathcal{T}')]&\leq \ln 2\cdot \opt+\alpha\cdot \opt+\frac{1}{1-\delta}\cdot\radlower\\
    & \leq (\ln 2+\alpha+\frac{1}{1-\delta})\cdot\opt,
\end{aligned}
\end{equation}
where the last inequality follows from Lemma \ref{radial lowerbound}. This finishes the proof of Theorem \ref{thm:LPtheorem}. We briefly comment that this algorithm can be derandomized efficiently using the method of conditional expectation since the probability a big client is covered and its expected contribution to the $\delta$-tank upper bound can be computed efficiently even if some tours have been sampled or rejected so far. Note there is a numerical issue in that $\gamma \cdot x_T$ may not be a rational number, but this error can be absorbed in the $\frac{1}{1-\delta}$ part of the guarantee by choosing $\delta$ to be slightly smaller.

\subsection{Further Improvements using the Blauth, Traub, and Vygen Approach}
The approach in \cite{VygenIPCO} first considers that the classic 3.5-approximation is in fact better if $\radlower$ is smaller than $\opt$ by some constant factor, say $\radlower \leq (1-\delta') \cdot \opt$. In the other case $\radlower > (1-\delta') \cdot \opt$, they show in fact that one can compute a TSP tour
of cost close to $\opt$ if $\delta'$ is sufficiently small. In particular, they show the following
in their full paper.

\begin{theorem}[Blauth, Traub, and Vygen \cite{VygenIPCO}, Theorem 23]\label{thm:vygen}
There is a function $f: \mathbb R_{> 0} \rightarrow \mathbb R_{> 0}$ with $\lim_{\delta' \rightarrow 0} f(\delta') = 0$ and a polynomial time algorithm for CVRP that, for any $\delta' > 0$, returns a solution of cost $(3 + f(\delta')) \cdot \opt$ for any instance with $\radlower \geq (1-\delta') \cdot \opt$.
\end{theorem}

They show taking the better of the standard 3.5-approximation for CVRP, which really finds a solution with cost $\alpha \cdot \opt + 2 \cdot \radlower$, and the algorithm from Theorem \ref{thm:vygen} produces a solution with cost $(\alpha + 2 \cdot (1 - \epsilon))\cdot \opt$ for some absolute constant $\epsilon > 0$.

This extends to our setting as well. If $\radlower \leq (1-\delta') \cdot \opt$ then both of our algorithms perform better than the bound we provide by a small factor of $\opt$. Otherwise, we can use the algorithm in \cite{VygenIPCO} to get a solution with cost $(3 + f(\delta'))\cdot \opt$.
Stating this formally, we have the following.

\begin{theorem}
Suppose there is an $\alpha$-approximation for TSP. Let $f$ be the function from Theorem \ref{thm:vygen}.
Then for any constant $\delta > 0$ there is an approximation for CVRP with unsplittable demands with guarantee $\min_{0 < \delta' < 1} \min\left\{3 + f(\delta'), \ln 2 + \alpha + \frac{1-\delta'}{1-\delta}\right\}$ with running time $n^{O(\frac{1}{\delta})}$.
\end{theorem}
\begin{proof}
The proof of Algorithm \ref{alg: 3.19} shows the cost of the solution produced is at most $(\ln 2 + \alpha) \cdot \opt + \frac{1}{1-\delta}\cdot \radlower$. For any $\delta'$, if $\radlower \leq (1-\delta') \cdot \opt$ then in fact Algorithm \ref{alg: 3.19} is a $\left(\ln 2 + \alpha + \frac{1-\delta'}{1-\delta}\right)-approximation$. Otherwise, if $\radlower > (1-\delta') \cdot \opt$ then the algorithm from Theorem \ref{thm:vygen} is a $(3 + f(\delta'))$-approximation. This holds for any 
$0 < \delta' < 1$.
\end{proof}

\begin{corollary}
Suppose for some constant $\alpha > 2 - \ln 2$ there is an $\alpha$-approximation for TSP\footnote{Any $\alpha \geq 4/3$ meets this criteria.}.
Then there is a constant $\epsilon' > 0$ such that CVRP with unsplittable demands admits a $(\ln 2 + \alpha + 1 - \epsilon')$-approximation.
\end{corollary}
\begin{proof}
For small enough $\delta'$, we have $3 + f(\delta') < \ln 2 + \alpha + 1$. Fix such a $\delta'$, then we can find $\delta$ small enough as well so that $\ln 2 + \alpha + \frac{1-\delta'}{1-\delta}$ is also smaller than $\ln 2 + \alpha + 1$. Thus, using this $\delta$ and taking the better of the two algorithms produces a solution whose cost is a constant-factor smaller than $(\ln 2 + \alpha + 1) \cdot \opt$, as required.
\end{proof}
We have not computed the exact improvement to our approximation guarantees, it seems to be in the order of $10^{-3}$ as in \cite{VygenIPCO}.

%Thus, taking the better of our algorithm and their algorithm would improve our ratios very slightly. We have not computed the exact improvement to our approximation guarantees, it seems to be in the order of $10^{-3}$.

% We can further improve this approximation factor by $1.5\cdot\epsilon$ where $\epsilon=\frac{1}{3000}$ is the constant in \cite{Vygen}, using a similar argument as in the end of Section \ref{sec: simple alg}.